\documentclass[10pt]{article}
\usepackage{fullpage}
\newif\ifacmstyle
\acmstylefalse
\usepackage{lmodern} 
\usepackage[T1]{fontenc} 
\usepackage[utf8]{inputenc}
\usepackage{hyperref} 
\usepackage{mdwlist} 
\usepackage[numbers,sort&compress,square]{natbib}
\ifacmstyle
\renewcommand{\refname}{REFERENCES}

\fi
\usepackage[ruled,linesnumbered]{algorithm2e} 
\usepackage{amssymb}
\usepackage{graphicx}
\usepackage{url}
\usepackage{color}
\usepackage{soul}
\usepackage{mathtools}
\newcommand*{\defeq}{\coloneqq}
\usepackage{caption}
\usepackage{subcaption}
\usepackage{nicefrac}
\usepackage{floatrow}
\usepackage{xspace}
\usepackage[show]{notes-alt}
\usepackage{ctable}
\usepackage{booktabs}

\ifacmstyle
\else
\usepackage{amsthm}
\fi
\newtheorem{corollary}{Corollary}
\newtheorem{lemma}{Lemma}
\newtheorem{theorem}{Theorem}

\ifacmstyle
\else
\theoremstyle{definition}
\fi
\newtheorem{definition}{Definition}

\newcommand{\para}[1]{\noindent{\bf #1.}}

\newcommand*{\algonamebase}{wiggins}
\newcommand*{\algonamebasecaps}{WIGGINS}
\newcommand*{\algonamebasecap}{Wiggins}
\newcommand*{\algoname}{\textsc{\algonamebase}\xspace}
\newcommand*{\algonameapx}{\textsc{\algonamebase-apx}\xspace}
\newcommand*{\algonamemr}{\textsc{\algonamebase-mr}\xspace}
\newcommand*{\cost}{\mathsf{cost}}
\newcommand*\degree{\mathsf{deg}}
\newcommand*\expect{\mathbb{E}}
\newcommand*\family{\mathcal{F}}
\newcommand*\fresh{\mathsf{f}}
\newcommand*\Sample{\mathcal{I}}
\newcommand*\sched{\mathsf{p}}
\newcommand*\sys{\Gamma}

\begin{document}
\author{Ahmad Mahmoody\footnote{Contact author}, Matteo Riondato, Eli Upfal\\
Department of Computer Science -- Brown University \\
\url{{ahmad, matteo, eli}@cs.brown.edu}}

\title{\algonamebasecap: Detecting Valuable Information\\in Dynamic Networks Using Limited Resources}
\maketitle

\textit{``Have you found it, Wiggins?''}
\begin{flushright}
	[Sherlock Holmes in \emph{A Study in Scarlet}]
\end{flushright}

\begin{abstract}
	Detecting new information and events in a dynamic network by probing individual
	nodes has many practical applications: discovering new webpages,
	analyzing influence properties in network, and detecting failure propagation
	in electronic circuits or infections in public drinkable water systems. In
	practice, it is infeasible for anyone but the owner of the network (if
	existent) 
	to monitor all nodes at all times. In this
	work we study the constrained setting when the observer can only probe a
	small set of nodes at each time step to check whether new pieces of
	information (items) have reached those nodes.

	We formally define the problem through an infinite time generating process
	that places new items in subsets of nodes according to an unknown
	probability distribution. Items have an exponentially decaying
	\emph{novelty}, modeling their decreasing value. The observer uses a
	\emph{probing schedule} (i.e., a probability distribution over the set of
	nodes) to choose, at each time step, a small set of nodes to check for new
	items. The goal is to compute a schedule that minimizes the average
	novelty of undetected items. We present an algorithm, \algoname, to
	compute the optimal schedule through convex optimization, and then show how
	it can be adapted when the parameters of the problem must be learned or
	change over time. We also present a scalable variant of \algoname for the
	MapReduce framework. The results of our experimental evaluation on real
	social networks demonstrate the practicality of our approach.
\end{abstract}

\section{Introduction}\label{sec:introduction}
Many applications require the detection of events in a network as soon as they
happen or shortly thereafter, as the value of the information obtained by
detecting the events decays rapidly as time passes. For example, an emerging
trend in algorithmic stock trading is the use of automatic search through the
Web and social networks for pieces of information that can be used in trading
decisions before they appear in the more popular news
sites~\citep{Delaney2009,latar2015robot,wallstreet2015,McKinney2011}. Similarly,
intelligence, business and politics analysts are scanning online sources for new
information or rumors. While new items are often reblogged, retweeted, and
posted on a number of sites, it is sufficient to find an item once, as fast as
possible, before it loses its relevance or freshness. There is no benefit in seeing multiple
copies of the same news item or rumor. This is also the case when monitoring for
intrusions, infections, or defects in, respectively, a computer network, a
public water system, or a large electronic circuit.

Monitoring for new events or information is a fundamental search and detection
problem in a distributed data setting, not limited to social networks or graph
analysis. In this setting, the data is distributed among a large number of
nodes, and new items appear in individual nodes (for example, as the products of
processing the data available locally at the node), and may propagate (being
copied) to neighboring nodes on a physical or a virtual network. The goal is to
detect at least one copy of each new item as soon as possible. The search
application can access any node in the system, but it can only \emph{probe}
(i.e., check for new items on) a few nodes at a time. To minimize the time to
find new items, the search application needs to optimize the schedule of probing
nodes, taking into account (i) the distribution of copies of items among the
nodes (to choose which nodes to probe), and (ii) the decay of the items'
\emph{novelty} (or relevance/freshness) over time (to focus the search on most
relevant items). The main challenge is how to devise a good probing schedule in
the absence of prior knowledge about the generation and distribution of items in
the network.


\paragraph*{Contributions}
In this work we study the novel problem of computing an optimal node probing
schedule for detecting new items in a network under resource scarceness, i.e.,
when only a few nodes can be probed at a time. Our contributions to the study of
this problem are the following:

\begin{itemize*}
	\item We formalize a generic process that describes the creation and
		distribution of information in a network, and define the computational task of
		learning this process by probing the nodes in the network according to a
		schedule. The process and task are parametrized by the resource
		limitations of the observer and the decay rate of the novelty of
		items. We introduce a \emph{cost} measure to compare different
		schedules: the cost of a schedule is the limit of the average expected
		novelty of uncaught items at each time step. On the basis of
		these concepts,  we formally define the \emph{Optimal Probing Schedule
		Problem}, which requires to find the schedule with minimum cost.
	\item We conduct a theoretical study of the cost of a schedule, showing that
		it can \emph{be computed explicitly} and that it is a \emph{convex
		function} over the space of schedules. We then introduce
		\algoname,\footnote{In the Sherlock Holmes novel \emph{A study in
		scarlet} by A.~Conan Doyle, Wiggins is the leader of the ``Baker Street
		Irregulars'', a band of street urchins employed by Holmes as
		intelligence agents.} an algorithm to compute the optimal schedule by
		solving a constrained convex optimization problem through the use of an
		iterative method based on Lagrange multipliers.
	\item We discuss variants of \algoname for the realistic situation where
		the parameters of the process needs to be learned or can change over
		time. We show how to compute a schedule which is (probabilistically)
		guaranteed to have a cost very close to the optimal by only observing
		the generating process for a limited amount of time. We also present a
		MapReduce adaptation of \algoname to handle very large networks.
	\item Finally, we conduct an extensive experimental evaluation of \algoname
		and its variants, comparing the performances of the schedules it
		computes with natural baselines, and showing how it performs extremely
		well in practice on real social networks when using well-established
		models for generating new items (e.g., the independence cascade
		model~\citep{Kempe2003}).
\end{itemize*}

To the best of our knowledge, the problem we study is novel and we are the first
to devise an algorithm to compute an optimal schedule, both when the generating
process parameters are known and when they need to be learned. 

\para{Paper Organization} In Sect.~\ref{sec:prelims} we give introductory
definitions, and formally introduce the settings and the problem. We discuss
related works in Sect.~\ref{sec:related_work}. In Sect.~\ref{sec:method} we
describe our algorithm \algoname and its variants. The results of our
experimental evaluation are presented in Sect.~\ref{sec:exp}. We conclude by
outlining directions for future work in Sect.~\ref{sec:concl}.

\section{Problem Definition}\label{sec:prelims}
In this section we formally introduce the problem and define our goal.

Let $G=(V,E)$ be a graph with $|V|=n$ nodes. W.l.o.g. we let $V=[n]$. Let
$\family\subseteq 2^V$ be a collection of subsets of $V$, i.e., a collection of
sets of nodes. Let $\pi$ be a function from $\family$ to $[0,1]$ (not
necessarily a probability distribution). We model the generation and diffusion
of  information in the network by defining a \emph{generating
process} $\sys=(\family,\pi)$.  $\sys$ is a infinite discrete-time process
which, at each time step $t$, generates a collection of sets $\Sample_t\subseteq\family$ such
that each set $S\in\family$ is included in $\Sample_t$ with probability $\pi(S)$,
independently of $t$ and of other sets generated at time $t$ and at time $t'<t$.
For any $t$ and any $S\in\Sample_t$, the ordered pair $(t,S)$ represents an
\emph{item} - a piece of information that was generated at time
$t$ and reached \emph{instantaneously} the nodes in $S$. We choose to model the
diffusion process as instantaneous because this abstraction accurately models the view of an outside
\emph{resource-limited observer} that does not have the resources to monitor simultaneously all
the nodes in the network at the fine time granularity needed to observe the
different stages of the diffusion process.

\para{Probing and schedule} The observer can only monitor the network by
\emph{probing} nodes. Formally, by \emph{probing a node $v\in V$ at time $t$},
we mean obtaining the set $I(t,v)$ of items $(t',S)$ such
that $t'\le t$ and $v\in S$:\footnote{The set $S$ appears in the notation for an
	item only for clarity of presentation: we are \emph{not} assuming that when
we probe a node and find an item $(t,S)$ we obtain information about $S$.}
\[
	I(t,v)\defeq\{(t',S) ~:~ t'\le t, S\in\Sample_{t'}, v\in S\}\enspace.
\]
Let $U_t$ be the union of the sets $\Sample_{t'}$ generated by $\sys$
at any time $t'\le t$, and so $I(t,v)\subseteq U_t$.

%
We model the \emph{resource limitedness of the observer} through a constant,
user-specified, parameter $c\in\mathbb{N}$, representing the maximum number of
nodes that can be probed at any time, where probing a node $v$ returns the value $I(t,v)$.

The observer chooses the $c$ nodes to probe by following a schedule. In this work
we focus on \emph{memoryless} schedules, i.e., the choice of nodes to probe at
time $t$ is independent from the choice of nodes probed at any time $t'<t$.
More precisely, a \emph{probing $c$-schedule} $\sched$ is a probability
distribution on $V$. At each time $t$, the observer chooses a set $P_t$ of $c$
nodes to probe, such that $P_t$ is obtained through \emph{random sampling of $V$
without replacement according to $\sched$}, independently from $P_{t'}$ from
$t'< t$.

\para{Caught items, uncaught items, and novelty}
We say that an item $(t',S)$ is \emph{caught at time $t\ge t'$} iff
\begin{enumerate*}
	\item a node $v\in S$ is probed at time $t$; and
	\item no node in $S$ was probed in the interval $[t',t-1]$.
\end{enumerate*}

Let $C_t$ be the set of items caught by the observer at any time $t'\le t$. We
have $C_t\subseteq U_t$. Let $N_t= U_t\setminus C_t$ be the \emph{set of
uncaught items at time $t$}, i.e., items that were generated at any time
$t'\le t$ and have not been caught yet at time $t$. For any item $(t',S)\in
N_t$, we define the \emph{$\theta$-novelty of $(t',S)$ at time $t$} as
\[
	\fresh_\theta(t,t',S)\defeq\theta^{t-t'},
\]
where $\theta\in(0,1)$ is a user-specified parameter modeling how fast the
value of an item decreases with time if uncaught. Intuitively, pieces of
information (e.g., rumors) have high value if caught almost as soon as they have
appeared in the network, but their value decreases fast (i.e., exponentially) as
more time passes before being caught, to the point of having no value in the
limit.

\para{Load of the system and cost of a schedule} The set $N_t$ of uncaught items
at time $t$ imposes a \emph{$\theta$-load, $L_\theta(t)$, on the graph at time
$t$}, defined as the sum of the $\theta$-novelty at time $t$ of the items in
$N_t$:
\[
	L_\theta(t)\defeq\sum_{(t',S)\in N_t} \fresh_\theta(t,t',S)\enspace.
\]
The quantity $L_\theta(t)$ is a random variable, depending both on $\sys$ and on
the probing schedule $\sched$, and as such it has an expectation
$\expect[L_\theta(t)]$ w.r.t.~all the randomness in the system. The
\emph{$\theta$-cost of a schedule $\sched$} is defined as the limit, for
$t\rightarrow\infty$, of the average expected load of the system:
\begin{align*}
	\cost_\theta(\sched)&\defeq\lim_{t\rightarrow\infty}\frac{1}{t}\sum_{t'\le
	t}\expect[L_\theta(t)]\\
	&=\lim_{t\rightarrow\infty}\frac{1}{t}\sum_{t'\le
	t}\expect\left[\sum_{(t'',S)\in N_{t'}}\fresh_\theta(t',t'',S)\right]\enspace.
\end{align*}
Intuitively, the load at each time indicates the amount of novelty we did not
catch at that time, and the cost function measures the average of such loss over
time. The limit above always exists (Lemma~\ref{lem:explicit}).

We now have all the necessary ingredients to formally define the problem of
interest in this work.

\para{Problem definition} Let $G=(V,E)$ be a graph and $\sys=(\family,\pi)$ be a
generating process on $G$. Let $c\in\mathbb{N}$ and
$\theta\in(0,1)$. The \emph{$(\theta,c)$-Optimal Probing Schedule Problem}
($(\theta,c)$-OPSP) requires to find the \emph{optimal $c$-schedule} $\sched^*$,
i.e., the schedule with \emph{minimum} $\theta$-cost over the set $\mathsf{S}_c$
of $c$-schedules:
\[
	\sched^*=\arg\min_{\sched} \{\cost_\theta(\sched), \sched\in\mathsf{S}_c\}\enspace.
\]
Thus, the goal is to design a $c$-schedule that discovers the maximum number of items weighted by their
novelty value (which correspond to those generated most recently).
The parameter $\theta$ controls how fast the novelty of an item decays, and influences the choices of a schedule.
When $\theta$ is closed to $0$, items are relevant only for a few steps and the schedule
must focus on the most recently generated items, catching them as soon as they are generated (or at most shortly
thereafter). At the other extreme ($\theta\approx 1$), an optimal schedule
must maximizes the total number of discovered items, as their novelty decays very
slowly.

Viewing the items as ``information'' disseminated in the network,
an ideal schedule assigns higher probing probability  to nodes
that act as \emph{ information hubs}, i.e., nodes that receive a large number of
items. Thus, an optimal schedule $\sched^*$, identifies
information hubs among the nodes. This task (finding information hubs) can be
seen as the complement of the \emph{influence maximization
problem}~\citep{Kempe2003,Kempe2005}. In the influence maximization problem we
look for a set of nodes that \emph{generate} information that reach most nodes.
In the information hubs problem, we are interested in a set of nodes that
\emph{receive} the most of information, thus the most informative nodes for an
observer.

In the following sections, we may drop the specification of the parameters
from $\theta$-novelty, $\theta$-cost, $\theta$-load, and $c$-schedule, and
from their respective notation, as the parameters will be clear from the
context.

\section{Related Work}\label{sec:related_work}
The novel problem we focus on in this work generalizes and complements a number of problems
studied in the literature.

The  ``Battle of Water Sensor Network'' challenge~\citep{BWSN2008} motivated a
number of works on \emph{outbreak detection}: the goal is to optimally place
static or moving sensors in water networks to detect
contamination~\citep{Leskovec2007,Krause2008,Hart2010}. The optimization can be
done w.r.t.~a number of objectives, such as maximizing the probability of
detection, minimizing the detection time, or minimizing the size of the
subnetwork affected by the phenomena~\citep{Leskovec2007}. A related
work~\citep{AgumbeSuresh2012} considered sensors that are sent along fixed paths
in the network with the goal of gathering sufficient information to locate
possible contaminations. Early detection of contagious outbreaks by monitoring
the neighborhood (friends) of a randomly chosen node (individual) was studied
by~\citet{Christakis2010}.  \citet{Krause2009} present efficient schedules for
minimizing energy consumption in battery operated sensors, while other works
analyzed distributed solutions with limited communication capacities and
costs~\citep{Krause2011Kleinberg,Golovin2010,Krause2011}. 
In contrast, our work is geared to detection in huge but virtual networks such
as the Web or social networks embedded in the Internet, where it is possible to
``sense'' or probe (almost) any node at approximately the same cost. Still only
a restricted number of nodes can be probed at each steps but the optimization of
the probing sequence is over a much larger domain, and the goal is to identify
the outbreaks (items) regardless of their size and solely by considering their
interest value.


Our methods complement the work on \emph{Emerging Topic Detection} where the
goal is to identify emergent topics in a social network, assuming full access to
the stream of all postings. Providers, such as Twitter or Facebook, have an
immediate access to all tweets or postings as they are submitted to their
server~\cite{Cataldi2010,Mathioudakis2010}. Outside observers need an efficient
mechanism to monitor changes, such as the methods developed in this work.

\emph{Web-crawling} is another research area that study how to obtain the
most recent snapshots of the web. However, it differs from our model in two key
points: our model allows items to propagate their copies, and they will be
caught if any of their copies is discovered (where snapshots of a webpage belong
to that page only), and all the generated items should be discovered (and not
just the recent ones)~\cite{dasgupta2007discoverability,wolf2002optimal}.

The goal of {News and Feed Aggregation} problem is to capture updates in news websites (e.g. by RSS feeds)~\cite{survey-oita2011deriving,onlineRef-horincar2014online,adaptive-bright2006adaptive,fast-sia2007efficient}. Our model differs from that setting in that we consider copies of the same news in different web sites as equivalent and therefore are only interested in discovering one of the copies.

\section{The \algonamebasecaps{} Algorithm}\label{sec:method}
In this section we present the algorithm \algoname
(and its variants) for solving
the Optimal Probing Schedule Problem $(\theta,c)$-OPSP for generating process
$\sys=(\family,\pi)$ on a graph $G=(V,E)$.

We start by assuming that we have complete knowledge of $\sys$, i.e., we know
$\family$ and $\pi$. This strong assumption allows us to study the theoretical
properties of the cost function and motivates the design of our algorithm,
\algoname, to compute an optimal schedule. We then remove the assumption and show
how we can extend \algoname to only use a collection of observations from
$\sys$. Then we discuss how to recalibrate our algorithms when the parameters of
the process (e.g., $\pi$ or $\family$) change over time.  Finally, we show an
algorithm for the MapReduce framework that allows us to scale to large networks.

\subsection{Computing the Optimal Schedule}\label{sec:optimize}
We first conduct a theoretical analysis of the cost function $\cost_\theta$.

\paragraph{Analysis of the cost function}
Assume for now that we know $\sys$, i.e., we have complete knowledge of
$\family$ and $\pi$. Under this assumption, we can exactly compute the
$\theta$-cost of a $c$-schedule.

\begin{lemma}\label{lem:explicit}
Let $\sched=(\sched_1,\ldots,\sched_n)$ be a $c$-schedule. Then
\begin{equation}\label{eq:explicit}
	\cost(\sched) \defeq
	\lim_{t\rightarrow\infty}\frac{1}{t}\sum_{t'=0}^t\expect[L_\theta(t')] =
	\sum_{S\in \family} \frac{\pi(S)}{1- \theta(1-\sched(S))^c},
\end{equation}
where $\sched(S) = \sum_{v\in S} \sched_v$.
\end{lemma}

\begin{proof}
	Let $t$ be a time step, and consider the quantity $\expect[L_\theta(t)]$. By
	definition we have
	\[
		\expect[L_\theta(t)]=\expect\left[\sum_{(t',S)\in
		N_t}\fresh_\theta(t,t',S)\right]=\expect\left[\sum_{(t',S)\in
		N_t}\theta^{t-t'}\right],
	\]
	where $N_t$ is the set of uncaught items at time $t$. Let now, for any
	$t'\le t$, $N_{t,t'}\subseteq N_t$ be the set of uncaught items in the form
	$(t',S)$. Then we can write
	\[
		\expect[L_\theta(t)]=\expect\left[\sum_{t'=0}^{t}\sum_{(t',S)\in
			N_{t,t'}}\theta^{t-t'}\right]\enspace.
	\]
	Define now, for each $S\in\family$, the random variable $X_{S,t',t}$ which
	takes the value $\theta^{t-t'}$ if $(t',S)\in N_{t,t'}$, and $0$ otherwise.
	Using the linearity of expectation, we can write:
	\begin{align}\label{eq:loadexp}
		\expect[L_\theta(t)]&=\sum_{S\in\family}\sum_{t'=0}^t\expect[X_S,t',t]\nonumber\\
		&=\sum_{S\in\family}\sum_{t'=0}^t
		\theta^{t-t'}\Pr(X_{S,t,t'}=\theta^{t-t'})\enspace.
	\end{align}

	The r.v. $X_{S,t,t'}$ takes value $\theta^{t-t'}$ if and only if the
	following two events $E_1$ and $E_2$ both take place:
	\begin{itemize*}
		\item $E_1$: the set $S\in F$ belongs to $\Sample_{t'}$, i.e., is
			generated by $\sys$ at time $t'$;
		\item $E_2$: the item $(t',S)$ is uncaught at time $t$. This is
			equivalent to say that no node $v\in S$ was probed in the time
			interval $[t',t]$.
	\end{itemize*}
	We have $\Pr(E_1)=\pi(S)$, and
	\[
		\Pr(E_2)=(1-\sched(S))^{c(t-t')}\enspace.
	\]
	The events $E_1$ and $E_2$ are independent, as the process of probing the
	nodes is independent from the process of generating items, therefore, we
	have
	\[
		\Pr(X_{S,t,t'}=\theta^{t-t'})=\Pr(E_1)\Pr(E_2)=\pi(S)(1-\sched(S))^{c(t-t')}\enspace.
	\]
	We can plug this quantity in the rightmost term of~\eqref{eq:loadexp} and
	write
	\begin{align}
		\lim_{t\rightarrow\infty}\expect[L_\theta(t)]
		&=\lim_{t\rightarrow\infty}\sum_{S\in\family}\sum_{t'=0}^t\theta^{t-t'}\pi(S)(1-\sched(S))^{c(t-t')}\nonumber\\
		&=\lim_{t\rightarrow\infty}\sum_{S\in\family}\pi(S)\sum_{t'=0}^t(\theta(1-\sched(S))^c)^{t}\nonumber\\
		&=\sum_{S\in\family}\frac{\pi(S)}{1-\theta(1-\sched(S))^c},
	\end{align}
	where we used the fact that $\theta(1-\sched(S))^c< 1$. We just showed that
	the sequence $(\expect[L_\theta(t)])_{t\in\mathbb{N}}$ converges as $t\rightarrow \infty$. Therefore,
	its Ces\`aro mean, i.e.,
	$\lim_{t\rightarrow\infty}\frac{1}{t}\sum_{t'=0}^t\expect[L_\theta(t)]$,
	equals to its limit~\citep[Sect.~5.4]{Hardy91} and we have
	\begin{align*}
		\cost_\theta(\sched)&=\lim_{t\rightarrow\infty}\frac{1}{t}\sum_{t'=0}^t\expect[L_\theta(t)] =
		\lim_{t\rightarrow\infty} \expect[L_\theta(t)]\\
		&=
		\sum_{S\in\family}\frac{\pi(S)}{1-\theta(1-\sched(S))^c}\enspace.
	\end{align*}
\end{proof}

We now show that $\cost_\theta(\sched)$, as expressed by the
r.h.s.~of~\eqref{eq:explicit} is a convex function over its domain
$\mathsf{S}_c$, the set of all possible $c$-schedules. We then use this result
to show how to compute an optimal schedule.

\begin{theorem}\label{thm:convexity}
	The cost function $\cost_\theta(\sched)$ is a convex function over
	$\mathsf{S}_c$.
\end{theorem}
\begin{proof}
	For any $S\in\family$, let
	\[
		f_S(\sched) = \frac{1}{1 - \theta(1-\sched(S))^c}\enspace.
	\]
	The function $\cost_\theta(\sched)$ is a linear combination of
	$f_S(\sched)$'s with positive coefficients. Hence to show that
	$\cost_\theta(\sched)$ is convex it is sufficient to show that, for any
	$S\in\family$, $f_S(\sched)$ is convex.

	We start by showing that $g_S(\sched) = \theta(1-\sched(S))^c$ is convex.
	This is due to the fact that its Hessian matrix is positive
	semidefinite~\citep{BoydV04}:
	\[
		\frac{\partial}{\partial \sched_i \partial \sched_j} g_S(\sched)=\left\{
		\begin{array}{lr}
		\theta c(c-1) (1-\sched(S))^{c-2} &  i,j \in S\\
		0 &  \text{otherwise}
		\end{array}
		\right.
	\]
	Let $\mathbf{v}_S$ be a $n\times 1$ vector in $\mathbb{R}^n$ such that its
	$i$-th coordinate is $\left[c(c-1) (1-\sched(S))^{c-2} \right]^{1/2}$ if $i\in
	S$, and $0$ otherwise. We can write the Hessian matrix of $g_S$ as
	\[
		\nabla^2 g_S = V_S * V_S^\mathsf{T},
	\]
	and thus, $\nabla^2 g_S$ is positive semidefinite matrix and $g$ is convex.
	From here, we have that $1-g_S$ is a \emph{concave} function. Since
	$f_S(\sched)=\frac{1}{1-g_S(\sched)}$ and the function $h(x)=\frac{1}{x}$ is
	convex and non-increasing, then $f_S$ is a convex function.
\end{proof}

If for every $v\in V$, $S=\{v\}$ belongs to $\family$, then the
function $g_S$ in the above proof is \emph{strictly} convex, and so is $f_S$.

We then have the following corollary of Thm.~\ref{thm:convexity}.

\begin{corollary}\label{corol:convexity}
	Any schedule $\sched$ with locally minimum cost is an optimal schedule
	(i.e., it has global minimum cost). Furthermore, if for every $v\in V$,
	$\{v\}$ belongs to $\family$, the optimal schedule is \emph{unique}.
\end{corollary}

\paragraph{The algorithm}
Corollary~\ref{corol:convexity} implies that one can compute an optimal
$c-$schedule $p^*$ (i.e., solve the $(\theta,c)$-OPSP) by solving the
unconstrained minimization of $\cost_\theta$ over the set $\mathsf{S}_c$ of all
$c$-schedules, or equivalently by solving the following constrained minimization
problem on $\mathbb{R}^n$:
\begin{equation}\label{eq:optimization}
	\boxed{
	\begin{array}{rrclcl}
	\displaystyle \min_{\sched\in\mathbb{R}^n} & \multicolumn{3}{l}{\cost_\theta(\sched)} \\
	&\displaystyle \sum_{i=1}^{n} \sched_i & = & 1 \\
	&\sched_i & \geq & 0 & & \forall i \in \{1,\ldots,n\} \\
	\end{array} }
\end{equation}
Since the function $\cost_\theta$ is convex and the constraints are linear, the
optimal solution can, theoretically, be found efficiently~\citep{BoydV04}. In
practice though, available convex optimization problem solvers can not scale
well with the number $n$ of variables, especially when $n$ is in the millions
as is the case for modern graphs like online social networks or the Web. Hence
we developed \algoname, an iterative method based on \emph{Lagrange
multipliers}~\citep[Sect.~5.1]{BoydV04}, which can scale efficiently and can be
adapted to the MapReduce framework of computation~\citep{dean2008mapreduce}, as
we show in Sect.~\ref{sec:mapreduce}. While we can not prove that this iterative
method always converges, we can prove (Thm.~\ref{thm:optimal}) that (i) if at
any iteration the algorithm examines an optimal schedule, then it will reach
convergence at the next iteration, and (ii) if it converges to a schedule, that
schedule is optimal. In Sect.~\ref{sec:exp} we show our experimental results
illustrating the convergence of \algoname in different cases.

\algoname takes as inputs the collection $\family$, the function $\pi$, and the
parameters $c$ and $\theta$, and outputs a schedule $\sched$ which, if
convergence (defined in the following) has been reached, is the optimal
schedule. It starts from a uniform schedule $\sched^{(0)}$, i.e.,
$\sched^{(0)}_i=1/n$ for all $1\le i\le n$, and iteratively refines it until
convergence (or until a user-specified maximum number of iterations have been
performed). At iteration $j\ge 1$, we compute, for each value $i$, $1\le i\le
n$, the function
\begin{equation}\label{eq:functionw}
	W_i(\sched^{(j-1)}) \defeq \sum_{\substack{S \in \family \\
		\mbox{\scriptsize{s.t. }} i\in
S}} \frac{\theta c \pi(S)
	(1-\sched^{(j-1)}(S))^{c-1}}{(1-\theta(1-\sched^{(j-1)}(S))^c)^2}
\end{equation}
and then set
\[
	\sched^{(j)}_i = \frac{\sched^{(j-1)}_i W_i(\sched^{(j-1)})}{\sum_{z=1}^n
	\sched^{(j-1)}_z W_z(\sched^{(j-1)})}\enspace.
\]
The algorithm then checks whether $\sched^{(j)}=\sched^{(j-1)}$. If so, then we
reached convergence and we can return $\sched^{(j)}$ in output, otherwise we
perform iteration $j+1$. The pseudocode for \algoname is in
Algorithm~\ref{alg:iterative}. The following theorem shows the correctness of
the algorithm in case of convergence.

\begin{theorem}\label{thm:optimal}
	We have that:
	\begin{enumerate*}
		\item if at any iteration $j$ the schedule $\sched^{(j)}$ is optimal,
			then \algoname reaches convergence at iteration $j+1$; and
		\item if \algoname reaches convergence, then the returned schedule
			$\sched$ is optimal.
	\end{enumerate*}
\end{theorem}

\begin{proof}
	From the method of the Lagrange multipliers~\citep[Sect.~5.1]{BoydV04}, we
	have that, if a schedule $\sched$ is optimal, then there exists a value
	$\lambda\in\mathbb{R}$ such that $\sched$ and $\lambda$ form a solution to
	the following system of $n+1$ equations in $n+1$ unknowns:
	\begin{equation}\label{eq:lagrange}
		\nabla [\cost_\theta(\sched) + \lambda (\sched_1+\ldots+\sched_n-1)] = 0,
	\end{equation}
	where the gradient on the l.h.s.~is taken w.r.t.~(the components of)
	$\sched$ and to $\lambda$ (i.e., has $n+1$ components).

	For $1\le i\le n$, the $i$-th equation induced by~\eqref{eq:lagrange} is
	\[
		\frac{\partial}{\partial \sched_i} \cost_\theta(\sched) + \lambda = 0,
	\]
	or, equivalently,
	\begin{equation}\label{eq:lagrange_i}
		\sum_{\substack{S\in\family\\\mbox{s.t.} i\in S}} \frac{\theta c
			\pi(S) (1-\sched(S))^{c-1}}{(1-\theta(1-\sched(S))^c)^2} =
			\lambda\enspace.
	\end{equation}
	The term on the l.h.s.~is exactly $W_i(\sched)$.
	The $(n+1)$-th equation of the system~\eqref{eq:lagrange} (i.e., the one
	involving the partial derivative w.r.t.~$\lambda$) is
	\begin{equation}\label{eq:lagrange_l}
		\sum_{z=1}^n \sched_z = 1\enspace.
	\end{equation}

	Consider now the first claim of the theorem, and assume that we
	are at iteration $j$ such that $j$ is the minimum iteration index for which
	the schedule $\sched^{(j)}$ computed at the end of iteration $j$ is optimal.
	Then, for any $i$, $1\le i\le n$, we have
	\[
		W_i(\sched^{(j)})=\lambda
	\]
	because $\sched^{(j)}$ is optimal and hence all identities in the form
	of~\eqref{eq:lagrange_i} must be true. For the same reason,
	\eqref{eq:lagrange_l} must also hold for $\sched^{(j)}$.
	Hence, for any $1\le i\le n$, we can write the value $\sched^{(j+1)}_i$
	computed at the end of iteration $j+1$ as
	\[
		\sched^{(j+1)}_i = \frac{\sched^{(j)}_i W_i(\sched^{(j)})}{\sum_{z=1}^n
		\sched^{(j)}_z W_z(\sched^{(j)})}=
		\frac{\sched^{(j)}_i\lambda}{1\lambda}=\sched^{(j)}_i,
	\]
	which means that we reached convergence and \algoname will return
	$\sched^{(j+1)}$, which is optimal. 

	Consider the second claim of the theorem, and let $j$ be the
	first iteration for which $\sched^{(j)}=\sched^{(j-1)}$. Then we have, for
	any $1\le i\le n$,
	\[
		\sched^{(j)}_i = \frac{\sched^{(j-1)}_i W_i(\sched^{(j-1)})}{\sum_{z=1}^n
		\sched^{(j-1)}_z W_z(\sched^{(j-1)})} = \sched^{(j-1)}_i\enspace.
	\]
	This implies
	\begin{equation}\label{eq:identityw}
		W_i(\sched^{(j-1)}) = \sum_{z=1}^n \sched^{(j-1)}_z W_z(\sched^{(j-1)})
	\end{equation}
	and the r.h.s.~does not depend on $i$, and so neither does
	$W_i(\sched^{(j-1)})$. Hence we have
	$W_1(\sched^{(j-1)})=\dotsb=W_n(\sched^{(j-1)})$ and can
	rewrite~\eqref{eq:identityw} as
	\[
		W_i(\sched^{(j-1)}) =  \sum_{z=1}^n \sched^{(j-1)}_z
		W_i(\sched^{(j-1)}),
	\]
	which implies that the identity~\eqref{eq:lagrange_l} holds for
	$\sched^{(j-1)}$. Moreover, if we set
	\[
		\lambda = W_1(\sched^{(j-1)})
	\]
	we have that all the identities in the form of~\eqref{eq:lagrange_i} hold.
	Then, $\sched^{(j-1)}$ and $\lambda$ form a solution to the
	system~\eqref{eq:lagrange}, which implies that $\sched^{(j-1)}$ is optimal
	and so must be $\sched^{(j)}$, the returned schedule, as it is equal to
	$\sched^{(j-1)}$ because \algoname reached convergence.
\end{proof}

\begin{algorithm}[ht]
	\DontPrintSemicolon
	\SetKwInOut{Input}{input}
	\SetKwInOut{Output}{output}
	\SetKwComment{tcp}{//}{}
	\SetKw{KwBreak}{break}
	\Input{$\family$, $\pi$, $c$, $\theta$, and maximum number $T$ of iterations}
	\Output{A $c$-schedule $\sched$ (with globally minimum $\theta$-cost, in
	case of convergence)}
	\For{$i\leftarrow 1$ \KwTo $n$}{
		$\sched_i \leftarrow 1/n$\;
	}
	\For{$j\leftarrow 1$ \KwTo $T$} {
		\For{$i\leftarrow 1$ \KwTo $n$}{
			$W_i\leftarrow 0$\;
		}
		\For{$S\in \family$} {\label{algline:sum}
			\For{$i\in S$} {
				$W_i \leftarrow W_i + \frac{\theta c \pi(S) (1-\sched(S))^{c-1}}{(1-\theta(1-\sched(S))^c)^2}$\label{algline:w}\;
			}
		}
		$\sched_{\mathrm{old}}\leftarrow \sched$\;
		\For{$i\leftarrow 1$ \KwTo $n$}{
			$\sched_i\leftarrow \frac{\sched_iW_i}{\sum_{i} \sched_iW_i}$\;
		}
		\If(// test for convergence){$\sched_{\mathrm{old}} = \sched$} {
			\KwBreak\;
		}
	}
	\Return{$\sched$}\;
	\caption{\algoname}
	\label{alg:iterative}
\end{algorithm}

\subsection{Approximation through Sampling}\label{sec:sampcomp}
We now remove the assumption, not realistic in practice, of knowing the
generating process $\sys$ exactly through $\family$ and $\pi$. Instead, we
observe the process using, for a limited time interval, a schedule that iterates
over all nodes (or a schedule that selects each node with uniform probability),
until we have observed, for each time step $t$ in a limited time interval
$[a,b]$, the set $\Sample_t$ generated by $\sys$, and therefore we have access
to a collection
\begin{equation}\label{eq:sample}
	\Sample=\{\Sample_a,\Sample_{a+1},\dotsc,\Sample_b\}.
\end{equation}
We refer to $\Sample$ as a \emph{sample gathered in the time interval
$[a,b]$}. We show that a schedule computed with respect to a sample $\Sample$ taken during an interval of
$ \ell (\Sample) = b-a=O(\varepsilon^{-2}\log n)$ steps has cost which is within  a
multiplicative factor $\varepsilon\in[0,1]$ of the optimal schedule. We then adapt \algoname to optimize with respect to such sample.


We start by defining the cost of a schedule w.r.t.~to a sample $\Sample$.

\begin{definition}\label{def:costsample}
	Suppose $\sched$ is a $c$-schedule and $\Sample$ is as in Equation~\eqref{eq:sample},
	with $\ell (\Sample) = b-a$. The $\theta$-cost of $\sched$ w.r.t.~to $\Sample$
	denoted by $\cost_\theta(\sched,\Sample)$ is defined as
	\[
		\cost_\theta(\sched,\Sample)\defeq\frac{1}{\ell(\Sample)}\sum_{S\in \Sample}
		\frac{1}{1-\theta(1-\sched(S))^c}\enspace.
	\]
\end{definition}

For $1\le i\le n$, define now the functions
\[
	W_i(\sched,\Sample) = \frac{1}{\ell(\Sample)}\sum_{S\in \Sample: i\in S}
	\frac{\theta c (1-\sched(S))^{c-1}}{(1-\theta(1-\sched(S))^c)^2}\enspace.
\]

We can then define a variant of \algoname, which we call \algonameapx. The
differences from \algoname are:
\begin{enumerate*}
	\item the loop on line~\ref{algline:sum} in Alg.~\ref{alg:iterative} is only
		over the sets that appear in at least one $\Sample_j\in\Sample$.
	\item \algonameapx uses the values $W_i(\sched,\Sample)$ (defined above) instead of
		$W_i(\sched)$ (line~\ref{algline:w} in Alg.~\ref{alg:iterative});
\end{enumerate*}

If \algonameapx reaches convergence, it returns a schedule with the minimum cost
w.r.t.~the sample $\Sample$. More formally, by following the same steps as in
the proof of Thm.~\ref{thm:optimal}, we can prove the following result about
\algonameapx.

\begin{lemma}\label{lem:optimal_sample}
	We have that:
	\begin{enumerate*}
		\item if at any iteration $j$ the schedule $\sched^{(j)}$ has minimum
			cost w.r.t.~$\Sample$, then \algonameapx reaches convergence at
			iteration $j+1$; and
		\item if \algonameapx reaches convergence, then the returned schedule
			$\sched$ has minimum cost w.r.t.~$\Sample$.
	\end{enumerate*}
\end{lemma}

Let $\ell(\Sample)$ denote the length of the time interval during which
$\Sample$ was collected. For a $c$-schedule $\sched$,
$\cost_\theta(\sched,\Sample)$ is an approximation of $\cost_\theta(\sched)$,
and intuitively the larger $\ell(\Sample)$, the better the approximation.

We now show that, if $\ell(\Sample)$ is large enough, then, with high
probability (i.e., with probability at least $1-1/n^r$ for some constant $r$),
the schedule $\sched$ returned by \algonameapx in case of convergence has a cost
$\cost_\theta(\sched)$ that is close to the cost $\cost_\theta(\sched^*)$ of an optimal
schedule $\sched^*$.

\begin{theorem}\label{thm:approx_sample}
	Let $r$ be a positive integer, and let $\Sample$ be a sample gathered during a
	time interval of length
	\begin{equation}\label{eq:samplesize}
		\ell(\Sample) \geq
		\frac{3(r\ln(n)+\ln(4))}{\varepsilon^2(1-\theta)}\enspace.
	\end{equation}
	Let $\sched^*$ be an optimal schedule, i.e., a schedule with minimum cost. If \algonameapx converges, then the returned schedule
	$\sched$ is such that
	\[
		\cost_\theta(\sched^*)\le\cost_\theta(\sched)\le\frac{1+\varepsilon}{1-\varepsilon}\cost_\theta(\sched^*)\enspace.
	\]
\end{theorem}

To prove Thm.~\ref{thm:approx_sample}, we need the following technical lemma.

\begin{lemma}\label{lem:chernoffcost}
	Let $\sched$ be a $c$-schedule and $\Sample$ be a sample gathered during a
	time interval of length
	\begin{align}\label{eq:samp_cond}
		\ell(\Sample) \geq \frac{3(r\ln(n)+\ln(2))}{\varepsilon^2(1-\theta)},
	\end{align}
	where $r$ is any natural number. Then, for every schedule $\sched$ we have
	\[
		\Pr(|\cost_\theta(\sched,\Sample) - \cost_\theta(\sched)|\geq
		\varepsilon\cdot\cost_\theta(\sched)) < \frac{1}{n^r}\enspace .
	\]
\end{lemma}

\begin{proof}
	For any $S\in\family$, let $X_S$ be a random variable which is
	$\frac{1}{1-\theta(1-\sched(S))^c}$ with probability $\pi(S)$, and zero
	otherwise. Since $\sched(S)\in [0,1]$, we have
	\[
		1\leq X_S \leq \frac{1}{1-\theta}\enspace.
	\]
	If we let $X = \sum_{S\in \family} X_S$, then
	\begin{equation}\label{eq:boundcost}
		\cost_\theta(\sched) = \expect[X] = \sum_{S\in\family} \expect[X_{S}]
		\geq \sum_{S\in\family}\pi(S)\enspace.
	\end{equation}
	Let $Z=\sum_{S\in\family}\pi(S)$. Then we have
	\[
		Z \leq X \leq \frac{Z}{1-\theta}\enspace.
	\]
	Let $X^i_S$ be the $i$-th draw of $X_S$, during the time interval $\Sample$ it
	was sampled from, and define $X^i = \sum_{S\in\family}X_S^i$. We have
	\[
		\cost_\theta(\sched,\Sample) =
		\frac{1}{\ell(\Sample)}\sum_{i}X^i\enspace.
	\]
	Let now
	\[
		\mu=\frac{\ell(\Sample)(1-\theta)}{|Z|}\cost_\theta(\sched)\enspace.
	\]
	By using the Chernoff bound for Binomial random
	variables~\citep[Corol.~4.6]{MitzenmacherU05}, we have
	\begin{align*}
		&\Pr\left(\left|\cost_\theta(\sched,\Sample) - \cost_\theta(\sched)\right| \geq \varepsilon
		\cost_\theta(\sched)\right)\\
		&=\Pr\left(\left|\sum_i X^i - \ell(\Sample)\cost_\theta(\sched)\right| \geq \varepsilon
		\ell(\Sample) \cost_\theta(\sched)\right) \\
		&= \Pr\left(\left|\frac{1-\theta}{|Z|}\sum_i X^i - \mu\right| \geq \varepsilon
		\mu\right) \leq 2\exp\left(-\frac{\varepsilon^2\mu}{3}\right)\\
		&\le 2\exp\left(-\frac{\varepsilon^2 \ell(\Sample) (1-\theta)
			\cost_\theta(\sched)}{3|Z|}\right) \leq 2\exp\left(-\frac{\varepsilon^2 \ell(\Sample) (1-\theta)}{3}\right),
	\end{align*}
	where the last inequality follows from the rightmost inequality
	in~\eqref{eq:boundcost}. The thesis follows from our choice of
	$\ell(\Sample)$.
\end{proof}

We can now prove Thm.~\ref{thm:approx_sample}.
\begin{proof}[of Thm.~\ref{thm:approx_sample}]
	The leftmost inequality is immediate, so we focus on the one on the right.
	For our choice of $\ell(\Sample)$ we have, through the union bound, that,
	with probability at least $1-1/n^r$, at the same time:
	\begin{align}\label{eq:boundedcosts}
		(1-\varepsilon)\cost_\theta(\sched)&\leq \cost_\theta(\sched,\Sample) &\le
		(1+\varepsilon)\cost_\theta(\sched) & \mbox{, and}\nonumber\\
		(1-\varepsilon)\cost_\theta(\sched^*)&\leq \cost_\theta(\sched^*,\Sample) &\le
		(1+\varepsilon)\cost_\theta(\sched^*)&
	\end{align}
	Since we assumed that \algonameapx reached convergence when computing
	$\sched$, then Thm.~\ref{thm:approx_sample} holds, and $\sched$ is a
	schedule with minimum cost w.r.t.~$\Sample$. In particular, it must be
	\[
		\cost_\theta(\sched,\Sample)\le \cost_\theta(\sched^*,\Sample)\enspace.
	\]
	From this and~\eqref{eq:boundedcosts}, We then have
	\[
		(1-\varepsilon)\cost_\theta(\sched)\leq \cost_\theta(\sched,\Sample) \le
		\cost_\theta(\sched^*,\Sample)\le (1+\varepsilon)\cost_\theta(\sched^*)
	\]
	and by comparing the leftmost and the rightmost terms we get the thesis.
\end{proof}


\subsection{Dynamic Settings}\label{sec:dynamic}
In this section we discuss how to handle changes in the parameters $\family$ and
$\pi$ as the (unknown) generating process $\sys$  evolves over time.
The idea is to maintain an estimation $\tilde{\pi}(S)$ of $\pi(S)$ for each set $S\in \family$ that we discover in the probing process,
together with the last
time $t$ such that an item $(t,S)$ has been generated (and caught at a time
$t'>t$). If we have not caught an item in the form $(t'',S)$ in an interval
significantly longer than $1/\tilde{\pi}(S)$, then we assume that
the parameters of $\sys$ changed. Hence, we trigger the collection of a new
sample and compute a new schedule as described in Sect.~\ref{sec:sampcomp}.


Note that when we adapt our schedule to the new environment (using the most
recent sample) the system converges to its stable setting exponentially (in
$\theta$) fast. Suppose $L$ items have been generated since we detected the
change in the parameters until we adapt the new schedule. These items, if not
caught, loose their novelty exponentially fast, since after $t$ steps their
novelty is at most $L\theta^t$ and decreases exponentially. In our experiments
(Sect.~\ref{sec:exp}) we provide different examples that illustrate how the load
of the generating process becomes stable after the algorithm adapts itself to
the changes of parameters.

\subsection{Scaling up with MapReduce}\label{sec:mapreduce}
In this section, we discuss how to adapt \algonameapx to the
MapReduce framework~\citep{dean2008mapreduce}. We denote the resulting algorithm
as \algonamemr.

In MapReduce, algorithms work in \emph{rounds}. At each round, first a function
$\mathsf{map}$ is executed independently (and therefore potentially massively in
parallel) on each element of the input, and a number (or zero) key-value pairs
of the form $(k,v)$ are emitted. Then, in the second part of the round, the
emitted pairs are partitioned by key and elements with the same key are sent to
the same machine (called the \emph{reducer} for that key), where a function
$\mathsf{reduce}$ is applied to the whole set of received pairs, to emit the
final output.

Each iteration of \algonameapx is spread over two rounds of \algonamemr. At
each round, we assume that the current schedule $\sched$ is available to all
machines (this is done in practice through a distributed cache). In the first
round, we compute the values $p_iW_i$, $1\le i\le n$, in the second round these
values are summed to get the normalization factor, and in the third round the
schedule $\sched$ is updated. The input in the first round are the sets
$S\in\Sample$. The function $\mathsf{map}_1(S)$ outputs, a key-value pair
$(i,v_S)$ for each $i\in S$, with
\[
	v_S = \frac{\theta c
		(1-\sched(S))^{c-1}}{\ell(\Sample)(1-\theta(1-\sched(S))^c)^2}\enspace.
\]
The reducer for the key $i$ receives the pairs $(i,v_S)$ for each $S\in\Sample$
such that $i\in S$, and aggregates them to output the pair $(i,g_i)$, with
\[
g_i=\sched_i\sum v_S = \sched_i W_i\enspace.
\]
The set of pairs $(i, g_i)$, $1\le i\le n$ constitutes the input to the
next round. Each input pair is sent to the same reducer,\footnote{This step can
	be made more scalable through \emph{combiners}, an advanced MapReduce
feature.} which computes the
value
\[
	g = \sum_{i=1}^n g_i = \sum_{i=1}^n\sched_iW_i
\]
and uses it to obtain the new values $\sched_i = g_i / g$, for $1\le i\le n$.
The reducer then outputs $(i,\sched_i)$.  At this point, the new schedule is
distributed to all machines again and a new iteration can start.

The same results we had for the quality of the final schedule computed by
\algonameapx in case of convergence carry over to \algonamemr.

\section{Experimental Results}\label{sec:exp}
In this section we present the results of our experimental evaluation of
\algonameapx.

\para{Goals} First, we show that for a given sample $\Sample$, \algonameapx
converges quickly to a schedule $\sched^*$ that minimizes $\cost_\theta(\sched,
\Sample)$ (see Thm.~\ref{thm:optimal}). In particular, our experiments
illustrate that the sequence $\cost_\theta(\sched^{(1)},\Sample),
\cost_\theta(\sched^{(2)}, \Sample),
\ldots$ is descending and converges after few iterations. Next, we compare the
output schedule of \algonameapx to four other schedules: (i) uniform schedules,
(ii) proportional to out-degrees, (iii) proportional to in-degrees, and (iv) proportional to \emph{undirected} degrees, i.e., the number of incident edges.
Specifically, we compute the costs of these schedules according to a sample
$\Sample$ that satisfies the condition in Lemma~\ref{lem:chernoffcost} and
compare them.
Then, we consider a specific example for which we know
the \emph{unique} optimal schedule, and show that for larger samples
\algonameapx outputs a schedule closer to the optimal.
 Finally, we demonstrate how our method can adapt itself to the changes in the
network parameters.


\begin{table}[ht]
\centering
\resizebox{0.65\columnwidth}{!}{
\begin{tabular}[scale=0.5]{l r r c r}
\toprule
Datasets	& \#nodes	& \#edges & $(|V_{1K}|,|V_{500}|,|V_{100}|)$ & gen. rate\\
\midrule
Enron-Email	& 36692	& 367662&(9,23,517) & 7.22\\
Brightkite	&58228	&428156	&(2,7,399)	& 4.54\\
web-Notredame	&325729	&1497134& (43,80,1619)& 24.49\\
web-Google	&875713	&5105039&(134,180,3546)	& 57.86\\
\bottomrule
\end{tabular}
}
\caption{\scriptsize The datasets, corresponding statistics, and the rate of generating new items at each step.}\label{table:datasets}
\end{table}

\para{Environment and Datasets} We implemented \algonameapx in C++. The implementation of \algonameapx never loads the entire sample
to the main memory, which makes it very practical when using large samples. The
experiments were run on a Opteron 6282 SE CPU (2.6 GHz) with 12GB of RAM. We
tested our method on graphs from the SNAP
repository\footnote{\url{http://snap.stanford.edu}} (see
Table~\ref{table:datasets} for details). We always consider the graphs to be
directed, replacing undirected edges with two directed ones.

\para{Generating process}
The generating process $\sys=(\family,\pi)$ we use in our experiments (except
those in Sect.~\ref{sec:example}) simulates an \emph{Independent-Cascade (IC)
model}~\citep{Kempe2003}. Since explicitly computing $\pi(S)$ in this case does
not seem possible, we simulate the creation of items according to this model as
follows. At each time $t$, items are generated in two phases: a ``creation''
phase and a ``diffusion'' phase. In the creation phase, we
simulate the creation of ``rumors'' at the nodes: we flip a biased coin
for each node in the graph, where the bias depends on the out-degree of the
node. We assume a partition of the nodes into classes based on their
out-degrees, and, we assign the same head probability for the biased coins of
nodes in the same class, as shown in Table~\ref{tab:bias}. In
Table~\ref{table:datasets}, for each dataset we report the size of the classes
and the expected number of flipped coins with outcome head at each time
(rightmost column).
\begin{table}[ht]
\centering
\resizebox{0.55\columnwidth}{!}{
	\begin{tabular}[scale=0.5]{lcl}
		\toprule
		Class & Nodes in class & Bias\\
		\midrule
		$V_{1K}$ & $\{i\in V ~:~ \degree^+(i) \geq 1000\}$ & 0.1\\
		$V_{500}$ &  $\{i\in V ~:~ 500\leq \degree^+(i) < 1000\}$ & 0.05 \\
		$V_{100}$ & $\{i\in V ~:~ 100\leq \degree^+(i) < 500\}$  & 0.01 \\
		$V_{0}$   & $\{i\in V ~:~ \degree^+(i) < 100\}$ & 0.0\\
		\bottomrule
	\end{tabular} }
	\caption{\scriptsize Classes and bias for the generating process.}
	\label{tab:bias}
\end{table}
Let now $v$ be a node whose coin had outcome head in the most recent flip. In
the ``diffusion'' phase we simulate the spreading of the ``rumor'' originating
at $v$ through network according to the IC model, as follows. For each directed
edge $e=u\rightarrow w$ we fix a probability $p_e$ that a rumor that reached $u$
is propagated through this edge to node $w$ (as in IC model), and events for
different rumors and different edges are independent. Following the
literature~\cite{Kempe2003,Chen2009,Chen2010,jung2011irie,tang2014influence}, we
use $p_{u\rightarrow w} = \frac{1}{\degree^-(w)}$. If we denote with $S$ the
final set of nodes that the rumor created at $v$ reached during the (simulated)
diffusion process (which always terminates), we have that through this process
we generated an item $(t,S)$, without the need to explicitly define $\pi(S)$.



\subsection{Efficiency and Accuracy}
In Sect.~\ref{sec:optimize} we showed that when a run of \algonameapx converges
(according to a sample $\Sample$) the computed $c$-schedule is optimal with
respect to the sample $\Sample$ (Lemma~\ref{lem:optimal_sample}). In our first
experiment, we measure the rate of convergence and the execution time of
\algonameapx. We fix $\epsilon=0.1$, $\theta=0.75$, and consider
$c\in\{1,3,5\}$. For each dataset, we use a sample $\Sample$ that satisfies~\eqref{eq:samp_cond}, and run \algonameapx
for 30 iterations. Denote the schedule computed at round $i$ by $\sched^i$.
As shown in Figure~\ref{fig:conv}, the sequence of cost values of the schedules
$\sched^i$'s, $\cost_\theta(\sched^i,\Sample)$, converges extremely fast after few iterations.

\begin{table}[ht]
\centering
\resizebox{0.65\columnwidth}{!}{
\begin{tabular}[scale=0.5]{l r r r}
\toprule
Datasets	& $|\Sample|$ & avg. item size & avg. iter. time (sec)\\
\midrule
Enron-Email		&97309	&12941.33  & 	204.59\\
Brightkite		&63652	&17491.08  & 	144.35\\
web-Notredame	&393348	&183.75	   &  	10.24\\
web-Google		&998038	&704.74 	   &	121.88\\
\bottomrule
\end{tabular}
}
\caption{\scriptsize Sample size, average size of items in the sample, and the running time of each iteration in \algonameapx (for $c=1$).}\label{table:time}
\end{table}

For each graph, the size of the sample $\Sample$, the average size of sets in
$\Sample$, and the average time of each iteration is given in
Table~\ref{table:time}. Note that the running time of each iteration is a
function of both sample size and sizes of the sets (informed-sets) inside the
sample.


\begin{figure}[htbp]
\subcaptionbox*{}{\includegraphics[width=0.49\textwidth]{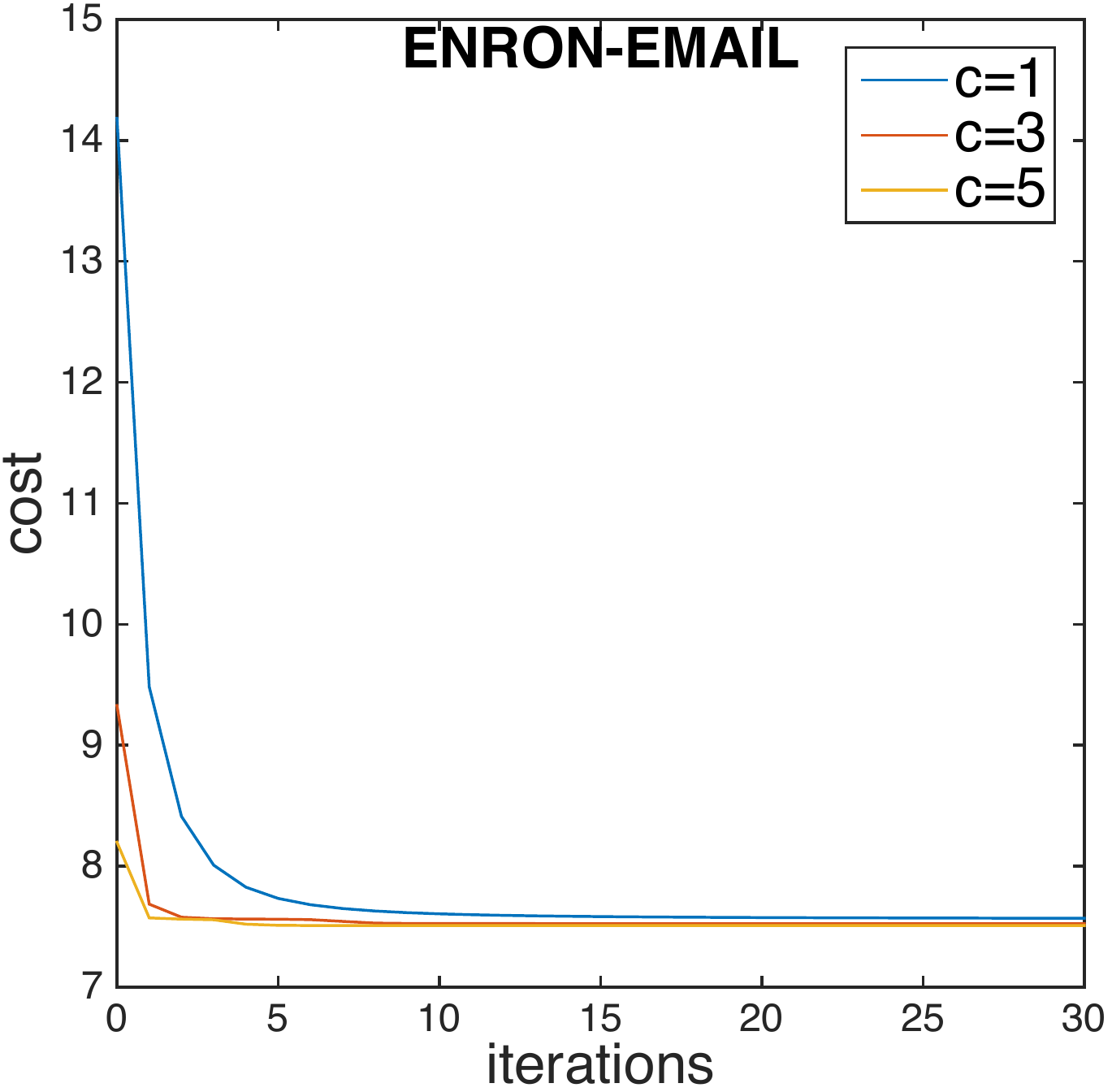}}
\subcaptionbox*{}{\includegraphics[width=0.49\textwidth]{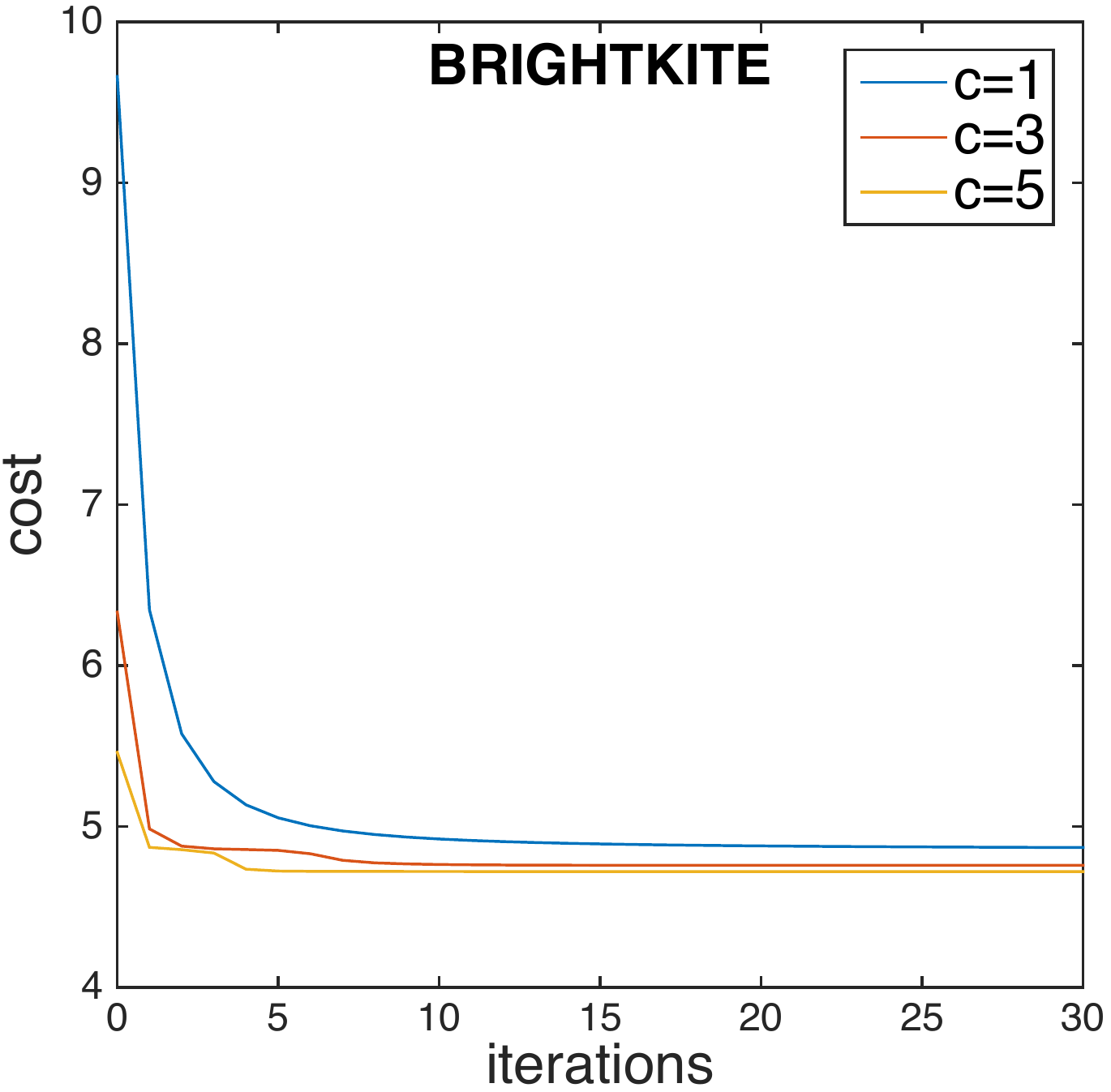}}
\subcaptionbox*{}{\includegraphics[width=0.49\textwidth]{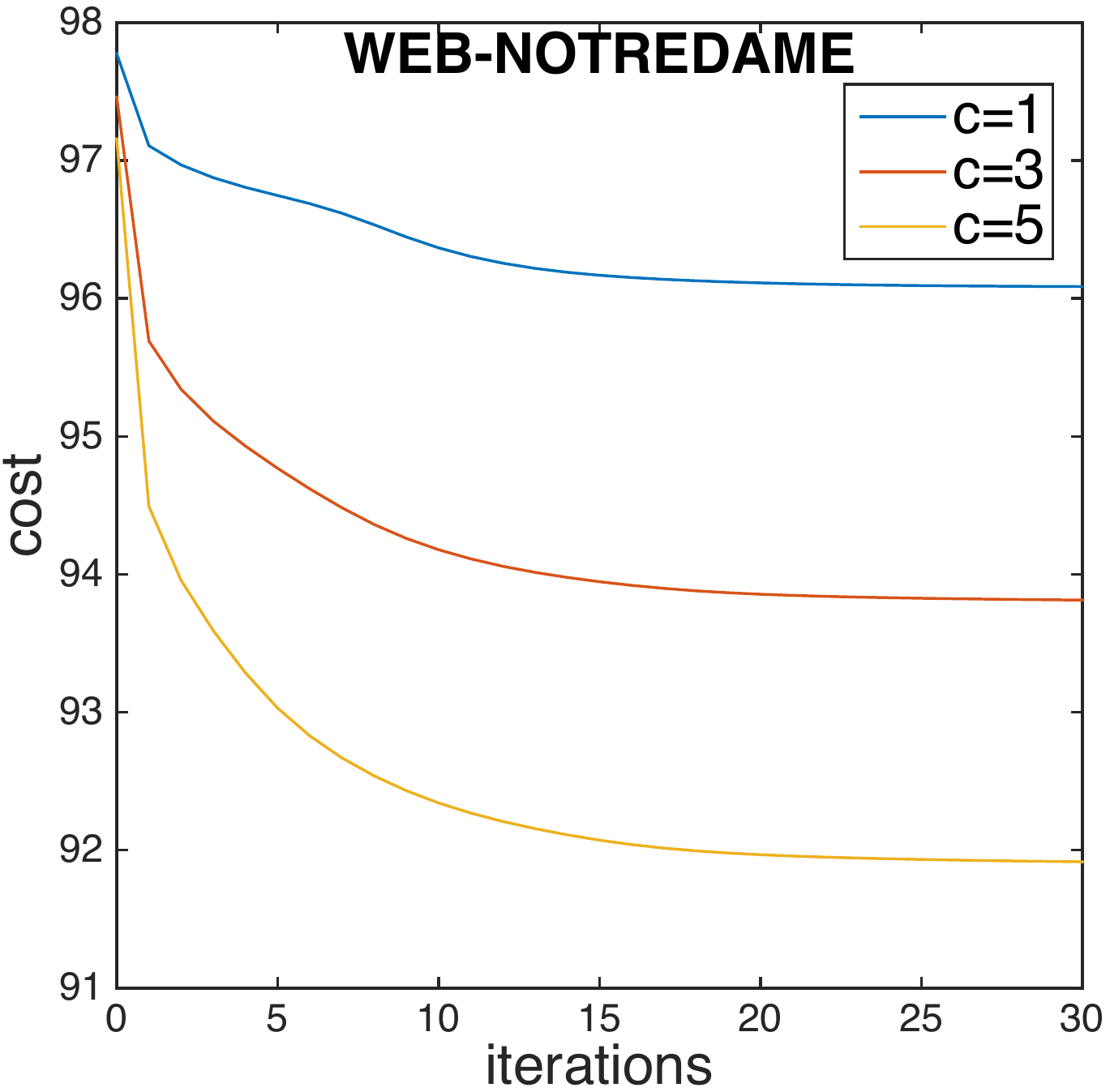}}
\subcaptionbox*{}{\includegraphics[width=0.49\textwidth]{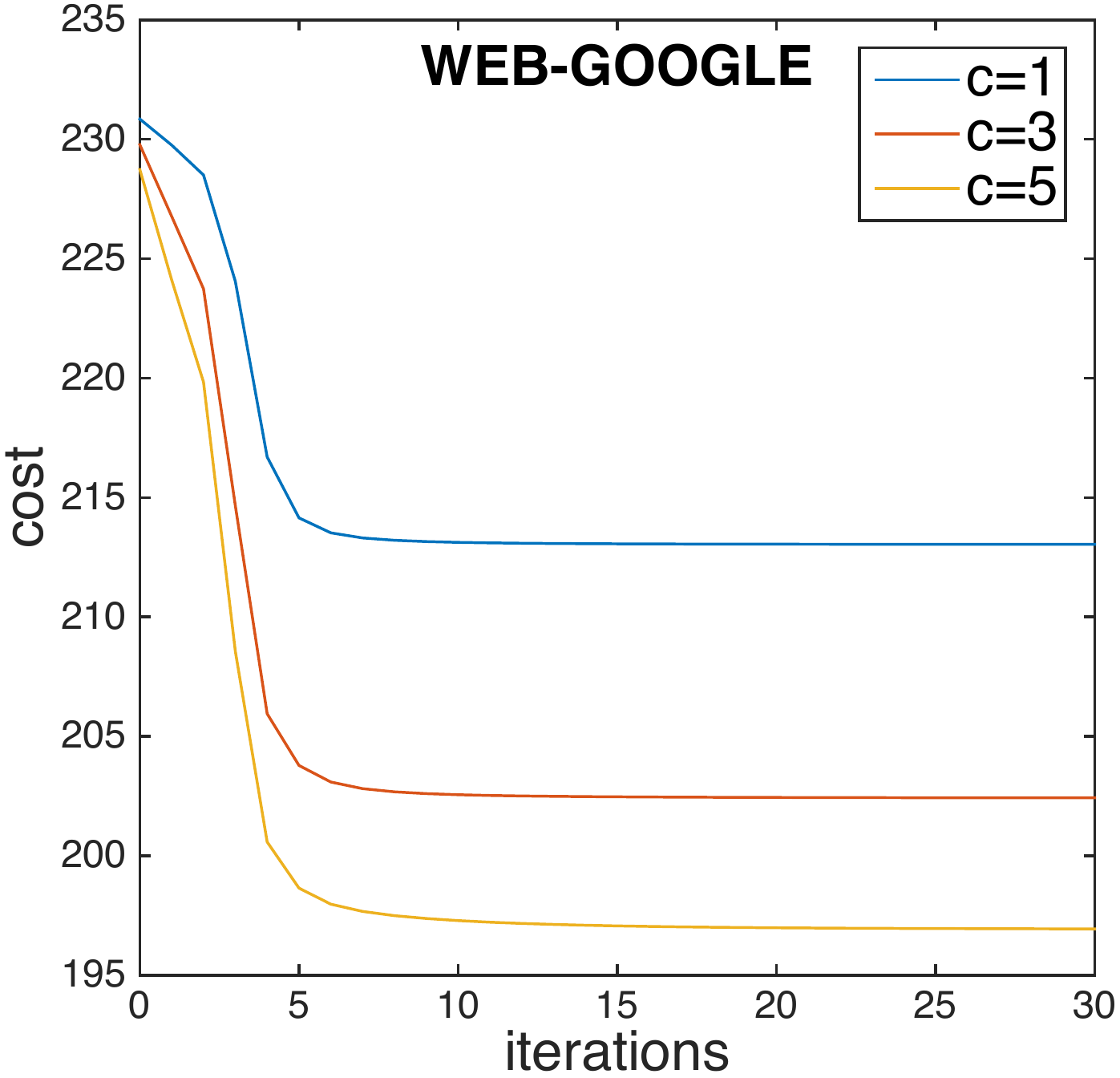}}
\caption{\scriptsize The cost of intermediate $c$-schedules at iterations of \algonameapx  according to $\Sample$.}\label{fig:conv}
\end{figure}

Next, we extract the $1$-schedules  output by \algonameapx, and compare its cost to four other natural schedules: \texttt{unif}, \texttt{outdeg}, \texttt{indeg}, and \texttt{totdeg} that probe each node, respectively, uniformly, proportional to its out-degree, proportional to its in-degree, and proportional to the number of incident edges. Note that for undirected graphs \texttt{outdeg}, \texttt{indeg}, and \texttt{totdeg} are essentially the same schedule.

To have a fair comparison among the costs of these schedules and \algonameapx, we calculate their costs according to 10  independent samples, $\Sample_1,\ldots,\Sample_{10}$ that satisfy \eqref{eq:samp_cond}, and compute the average. The results are shown in Table~\ref{table:compare}, and show that \algonameapx outperforms the other four schedules.


\begin{table}[ht]
\centering
\resizebox{0.65\columnwidth}{!}{
\caption{\scriptsize Comparing the costs of 5 different 1-schedules.}
\label{table:compare}
\begin{tabular}{lrrrrr}
\toprule
Dataset       & \algonameapx & \texttt{uniform} & \texttt{outdeg} & \texttt{indeg} & \texttt{totdeg} \\
\midrule
Enron-Email   &       7.55 & 14.16 & 9.21 & 9.21 & 9.21                  \\
Brightkite    &       4.85 & 9.64 & 6.14 & 6.14 & 6.14                  \\
web-Notredame &       96.10 & 97.78 & 97.37 & 97.43 & 97.40             \\
web-Google    &       213.15 & 230.88 & 230.48 & 230.47 & 230.47       \\
\bottomrule
\end{tabular}
}
\end{table}

\subsubsection{A Test on Convergence to  Optimal Schedule}\label{sec:example}
Here, we further invetigate the convergence of \algonameapx, using an example
graph and process for which we know the \emph{unique} optimal schedule. We study
how close the \algonameapx output is to the optimal schedule when (i) we start
from different initial schedules, $\sched^0$, or (ii) we use samples $\Sample$'s
obtained during time intervals of different lengths.

Suppose $G=(V,E)$ is the complete graph where $V=[n]$. Let $\sys=(\family, \pi)$ for  $\family = \{S \in 2^{[n]} \mid 1 \leq|S| \leq 2\}$, and $\pi(S)=\frac{1}{|\family|}$. It is easy to see that $\cost_\theta(\sched)$ is a symmetric function, and thus, the uniform schedule is optimal. Moreover, by Corollary~\ref{corol:convexity} the uniform schedule is the only optimal schedule, since $\{v\} \in \family$ for every $v\in V$. Furthermore, we let $\theta=0.99$ to increase the sample complexity (as in Lemma~\ref{lem:chernoffcost}) and make it harder to learn the uniform/optimal schedule.

\begin{figure}[htbp]
	\subcaptionbox*{}{\includegraphics[width=0.49\textwidth]{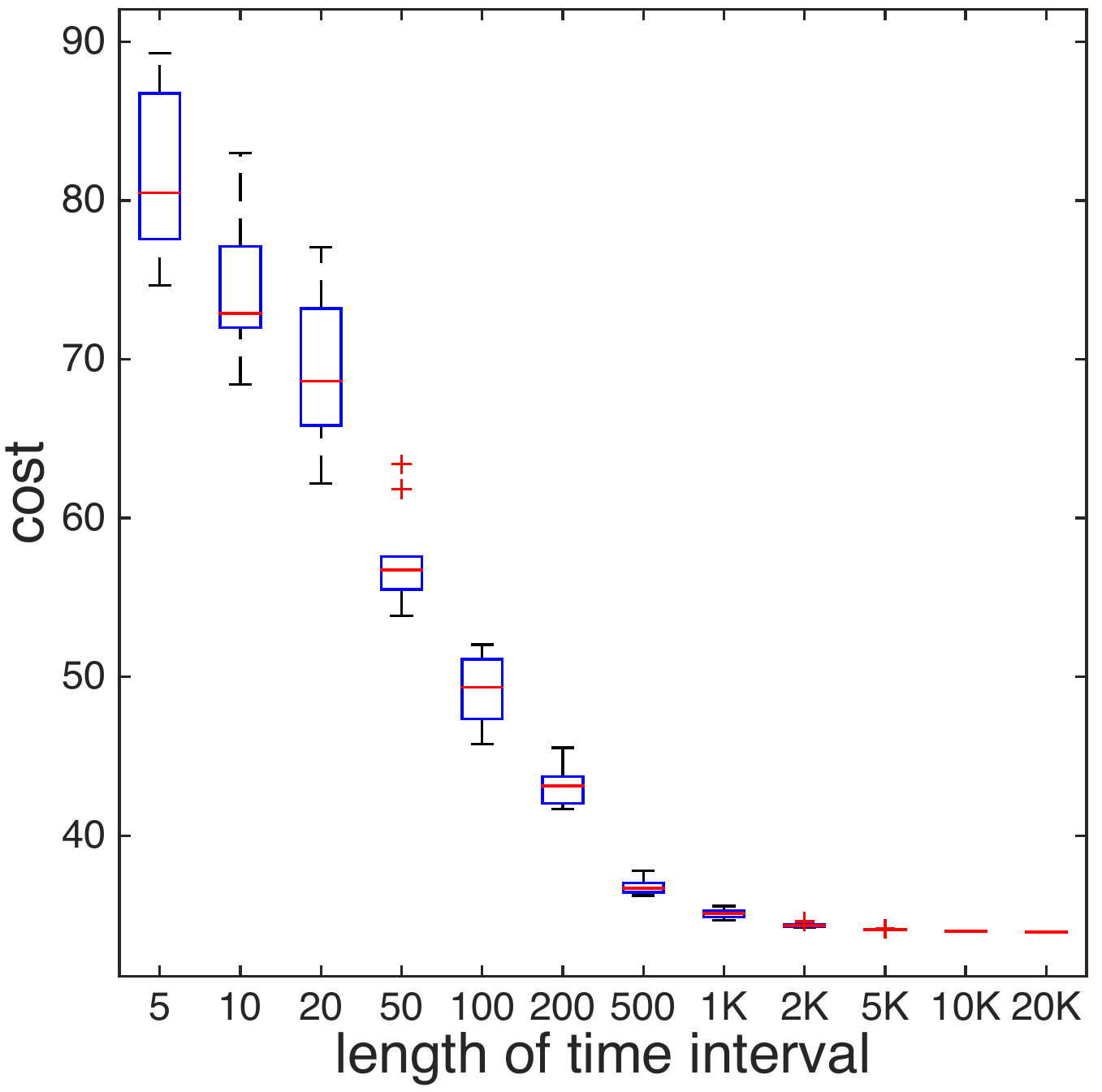}}
	\subcaptionbox*{}{\includegraphics[width=0.49\textwidth]{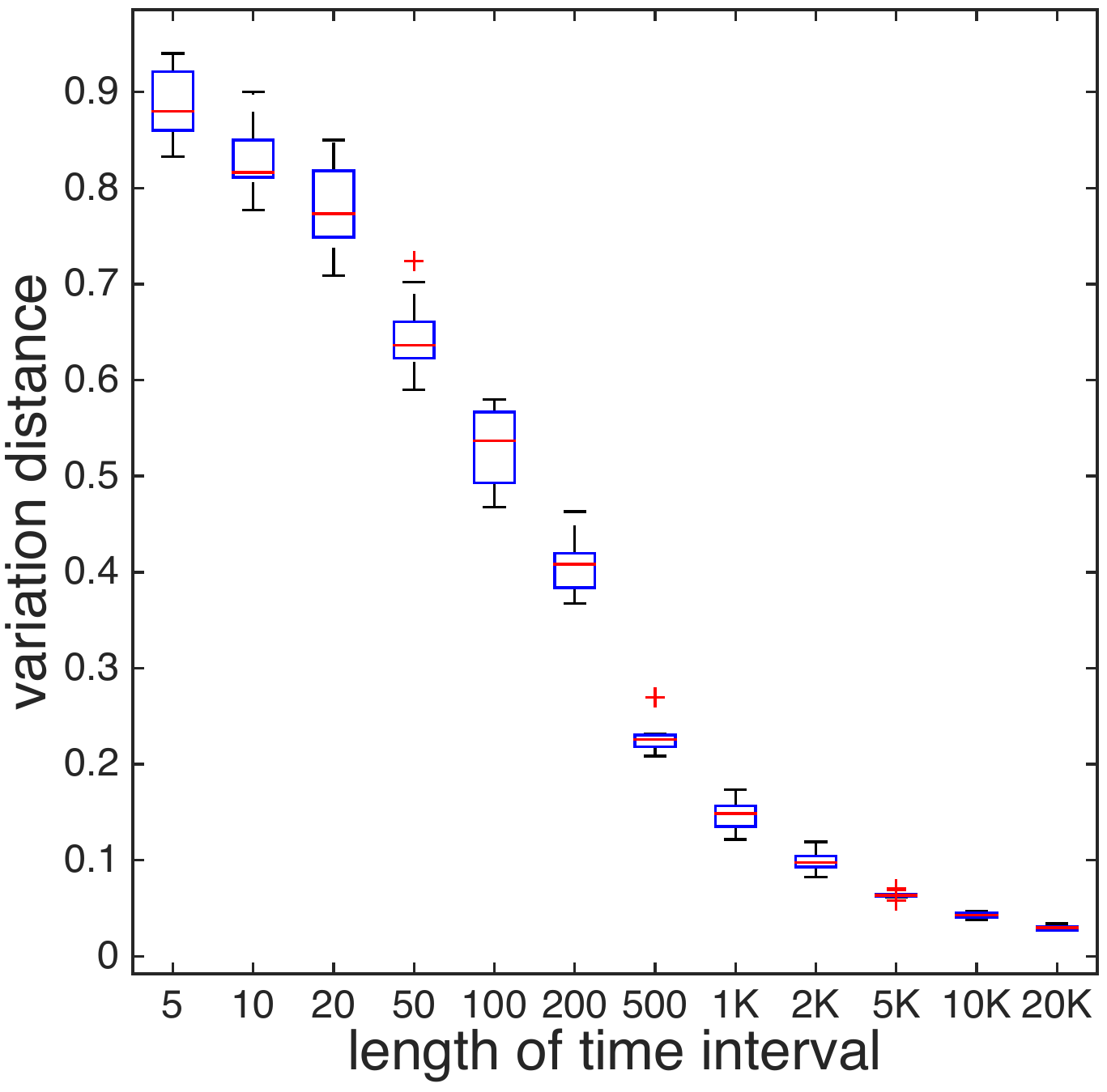}}
	\caption{\scriptsize The cost of \algonameapx outputs and their variation distance to the optimal schedule.} \label{fig:unique}
\end{figure}
In our experiments we run the \algonameapx algorithm, using (i) different \emph{random} initial schedules, and (ii) samples $\Sample$ obtained from time intervals of different lengths. For each sample, we run \algonameapx 10 times with 10 different random initial schedules, and compute the \textbf{exact} cost of each schedule, and its variation distance to the uniform schedule. Our results are plotted in Figure~\ref{fig:unique}, and as shown, by increasing the sample size (using longer time intervals of sampling) the output schedules gets very close to the uniform schedule (the variance gets smaller and smaller).


\subsection{Dynamic Settings}\label{sec:dynset}
In this section, we present experimental results that show how our algorithm can
adapt itself to the new situation. The experiment is illustrated in
Fig.~\ref{figure:changes}. For each graph, we start by following an optimal
1-schedule in the graph. At the beginning of each ``gray'' time interval, the
labels of the nodes are permuted randomly, to impose great disruptions in the
system. Following that, at
the beginning of each ``green'' time interval our algorithm starts gathering
samples of $\sys$. Then, \algonameapx computes the schedule for the new sample, using 50 rounds of iterations, and starts probing. The length of each colored time
interval is $R = \frac{3(\log(n)+\log(2))}{\epsilon^2(1-\theta)}$, $\epsilon=0.5$, motivated by
Theorem~\ref{thm:approx_sample}.


Since the cost function is defined asymptotically (and explains the asymptotic
behavior of the system in response to a schedule), in
Figure~\ref{figure:changes} we plot the load of the system $L_\theta(t)$ over
the time (blue), and the \emph{average} load in the normal and perturbed time
intervals (red). Based on this experiment, and as shown in
Figure~\ref{figure:changes}, after adapting to the new schedule, the effect of
the disruption caused by the perturbation disappears immediately. Note that when
the difference between the optimal cost and any other schedule is small (like
web-Notredame), the jump in the load will be small (e.g., as shown in
Figure~\ref{fig:conv} and Table~\ref{table:compare}, the cost of the initial
schedule for web-Notredame is very close the optimal cost, obtained after 30
iteration).


\begin{figure}
  \subcaptionbox*{}{\includegraphics[width=0.49\textwidth]{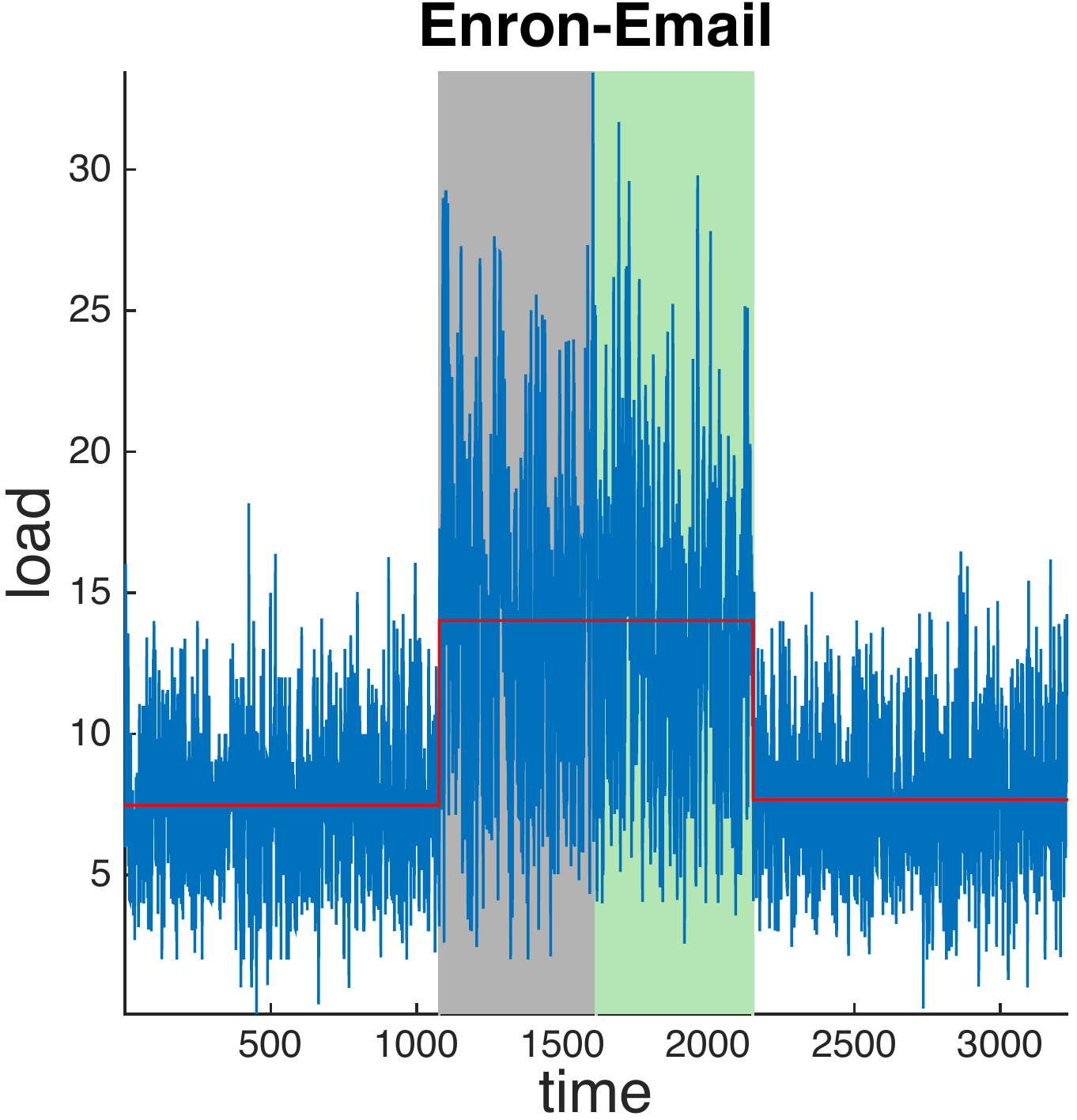}} 
  \subcaptionbox*{}{\includegraphics[width=0.49\textwidth]{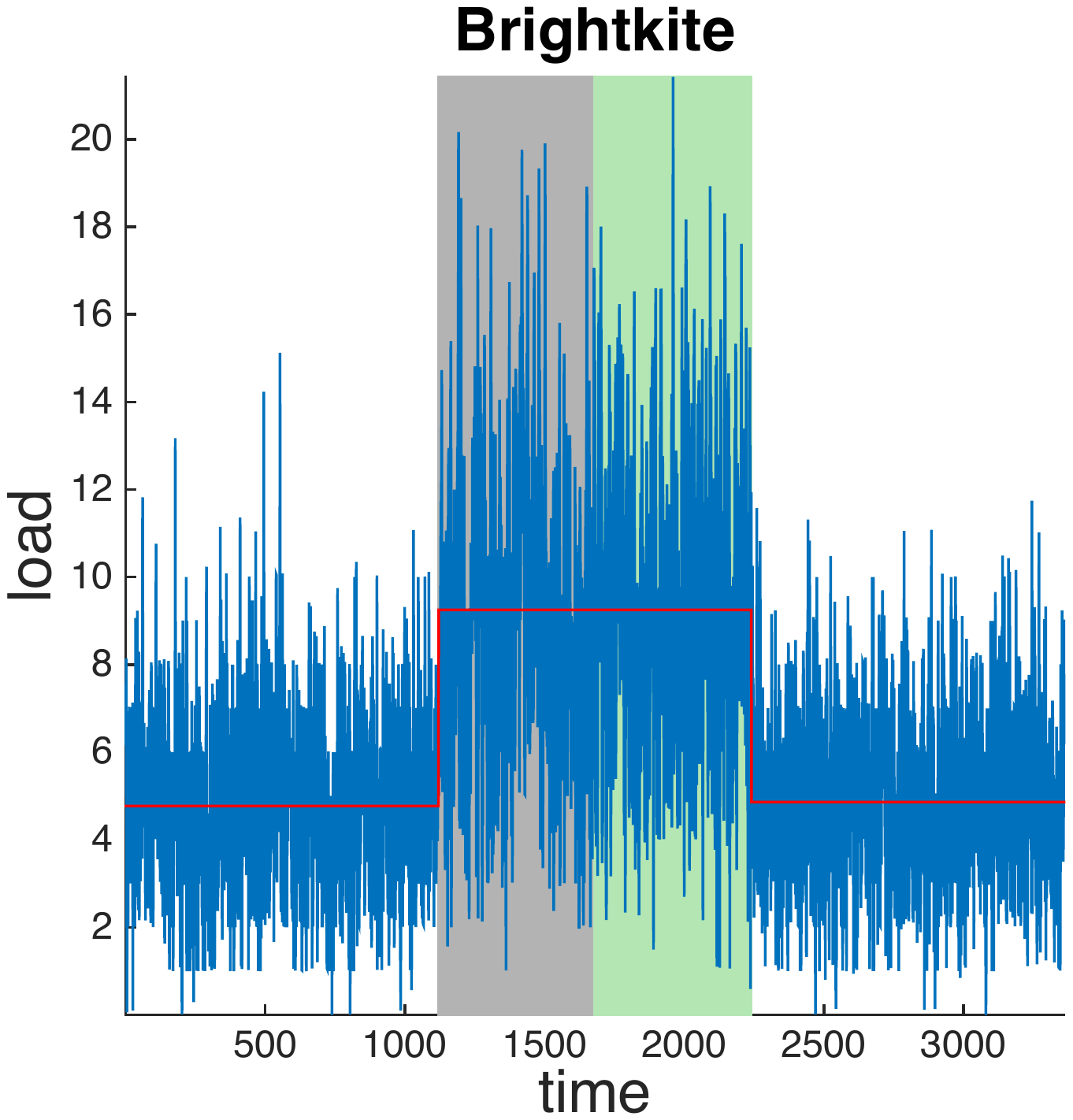}} 
  \subcaptionbox*{}{\includegraphics[width=0.49\textwidth]{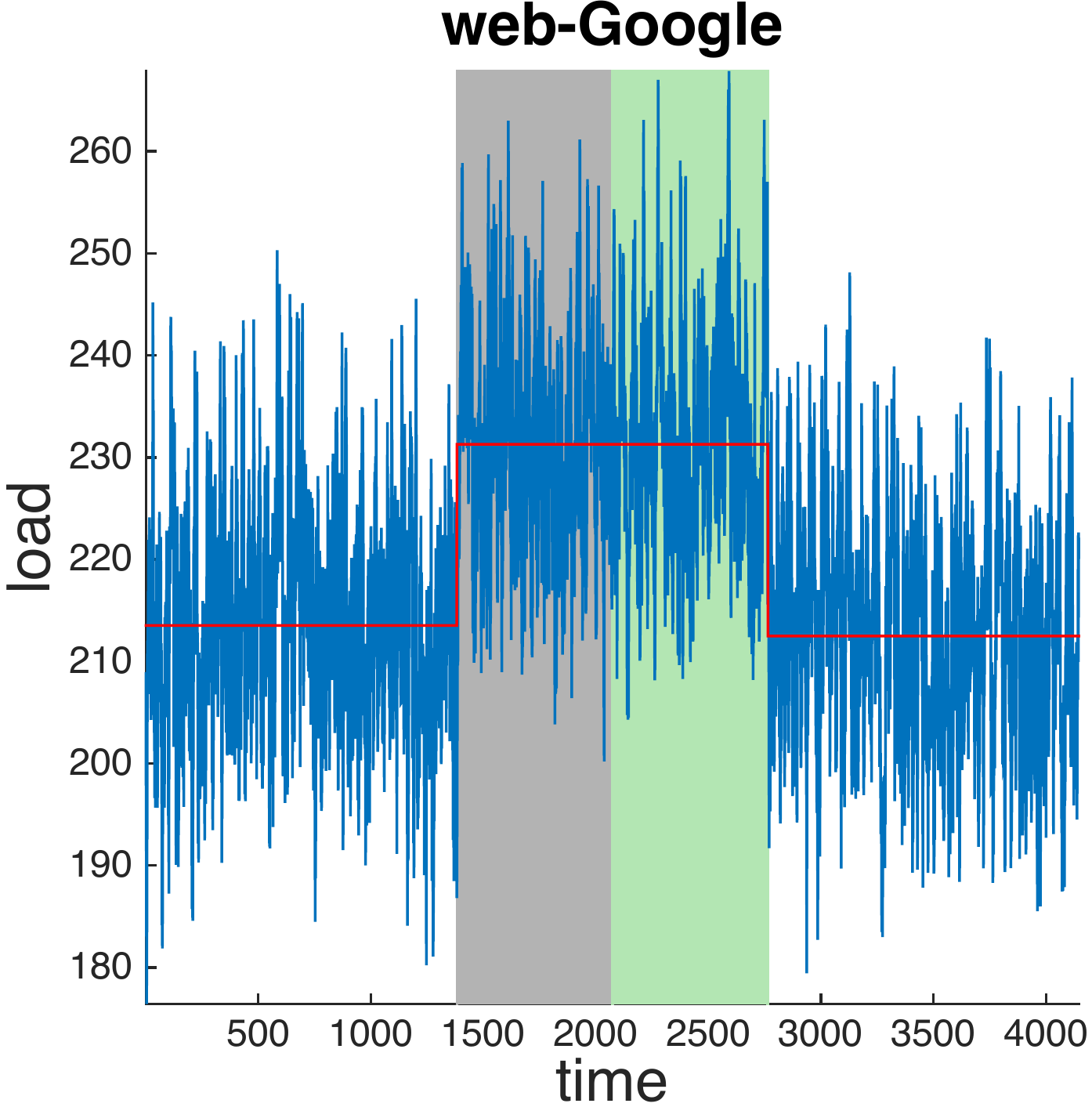}} 
  \subcaptionbox*{}{\includegraphics[width=0.49\textwidth]{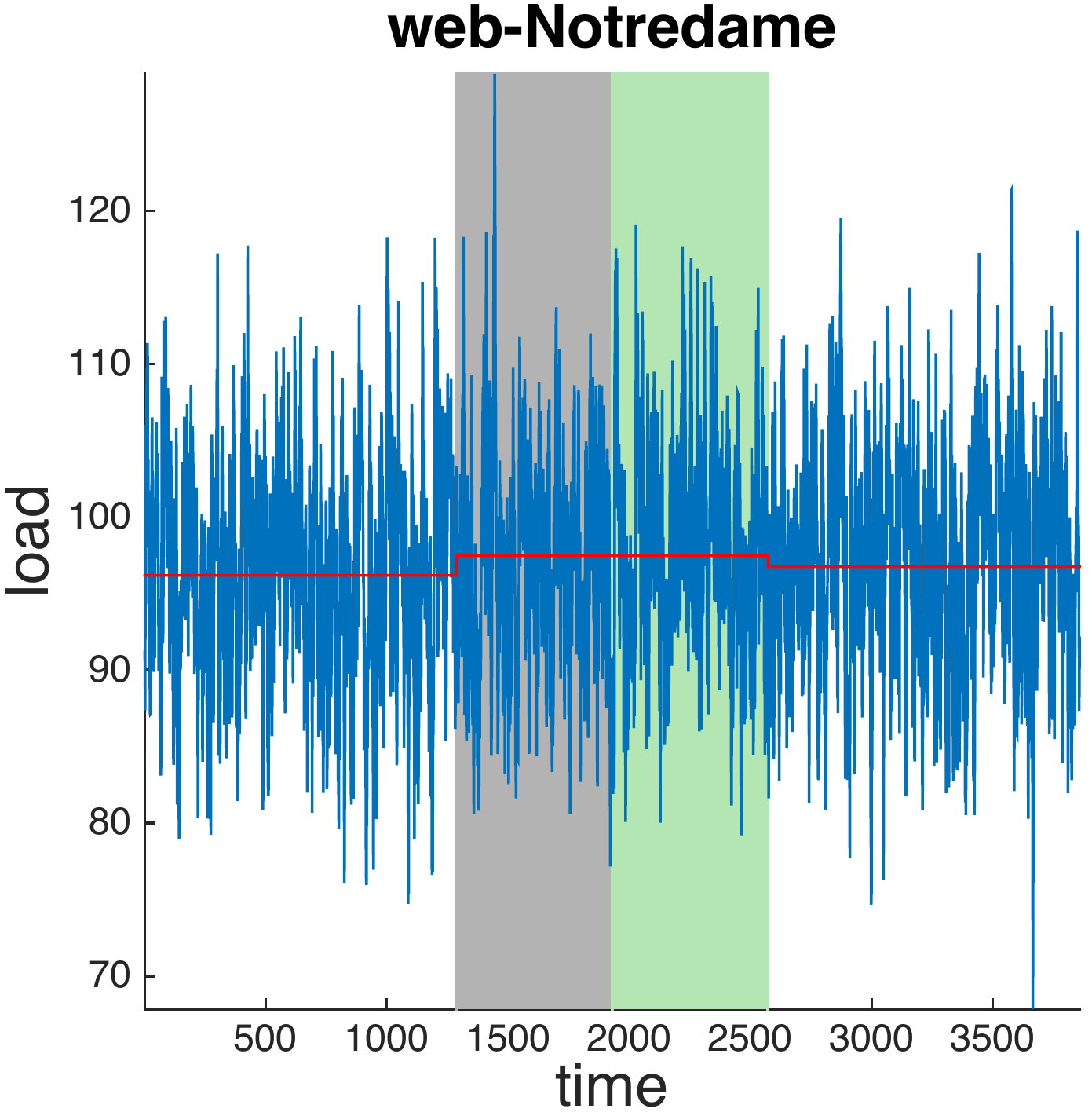}} 
  \caption{\scriptsize Perturbation, Sampling, and Adapting (For details see Section~\ref{sec:dynset}).}\label{figure:changes}
\end{figure}

\section{Conclusions}\label{sec:concl}
We formulate and study the $(\theta,c)$-Optimal Probing Schedule Problem,
which requires to find the best probing schedule that allows an observer to find
most pieces of information recently generated by a process $\sys$, by probing a
limited number of nodes at each time step.

We design and analyze an algorithm, \algoname, that can solve the problem
optimally if the parameters of the process $\sys$ are known, and then design a
variant that computes a high-quality approximation of the optimum schedule when
only a sample of the process is available. We also show that \algoname can be
adapted to the MapReduce framework of computation, which allows us to scale up
to networks with million of nodes. The results of experimental evaluation on a
variety of graphs and generating processes show that \algoname and its variants
are very effective in practice.

Interesting directions for future work include generalizing the problem to allow
for non-memoryless schedules and different novelty functions.

\section{Acknowledgements}
This work was supported by NSF grant IIS-1247581 and NIH grant R01-CA180776.

\bibliographystyle{abbrvnat}
\bibliography{catchRef}

\end{document}